\documentclass [14pt]{amsart}
\usepackage{amsmath}
\usepackage{amssymb}
\usepackage{amscd}
\usepackage{epsfig}
\usepackage{amsxtra}
\usepackage{pifont}
\usepackage{mathabx}
\usepackage{MnSymbol}
\usepackage{accents}
\usepackage{scalerel,stackengine}
\usepackage{xcolor}
\usepackage{tikz-cd}
\allowdisplaybreaks
\stackMath
\newcommand\reallywidehat[1]{%
\savestack{\tmpbox}{\stretchto{%
  \scaleto{%
    \scalerel*[\widthof{\ensuremath{#1}}]{\kern-.6pt\bigwedge\kern-.6pt}%
    {\rule[-\textheight/2]{1ex}{\textheight}}
  }{\textheight}%
}{0.5ex}}%
\stackon[1pt]{#1}{\tmpbox}%
}

\newcommand\reallywidecheck[1]{%
\savestack{\tmpbox}{\stretchto{%
  \scaleto{%
    \scalerel*[\widthof{\ensuremath{#1}}]{\kern-.6pt\bigwedge\kern-.6pt}%
    {\rule[-\textheight/2]{1ex}{\textheight}}
  }{\textheight}%
}{0.5ex}}%
\stackon[1pt]{#1}{\scalebox{-1}{\tmpbox}}%
}

\usetikzlibrary{decorations.markings}
\tikzset{double line with arrow/.style args={#1,#2}{decorate,decoration={markings,%
mark=at position 0 with {\coordinate (ta-base-1) at (0,1pt);
\coordinate (ta-base-2) at (0,-1pt);},
mark=at position 1 with {\draw[#1] (ta-base-1) -- (0,1pt);
\draw[#2] (ta-base-2) -- (0,-1pt);
}}}}

\numberwithin{equation}{section}

\newcommand{\supp}{\mbox{\rm supp}}

\newcommand{\RR}{{\mathbb R}}

\newcommand{\ZZ}{{\mathbb Z}}
\newcommand{\CC}{{\mathbb C}}

\newcommand{\NN}{\mathbb N}

\newcommand{\M}{{\mathcal M}}
\newcommand{\cL}{{\mathcal L}}
\newcommand{\cR}{{R}}

\newcommand{\oplam}{\mbox{\Large $\curlywedge$}}

\newcommand{\cM}{{\mathcal M}}

\newcommand{\cT}{{\mathcal T}}

\newcommand{\dd}{\,{\rm d}}

\newcommand{\Cu}{C_{\mathsf{u}}}
\newcommand{\Cc}{C_{\mathsf{c}}}
\newcommand{\Cb}{C_{\mathsf{b}}}

\newcommand{\SAP}{\mathcal{SAP}}
\newcommand{\GAP}{\mathcal{GAP}}

\newcommand{\vL}{\varLambda}
\newcommand{\vG}{\varGamma}
\newcommand{\vOmega}{\varOmega}
\newcommand{\udens}{{\overline{\mbox{\rm dens}}}}

\newcommand{\eps}{\varepsilon}
\newcommand{\card}{{\rm card}}

 \newtheorem{theorem}{Theorem}[section]
 \newtheorem{lemma}[theorem]{Lemma}
 \newtheorem{proposition}[theorem]{Proposition}
 \newtheorem{corollary}[theorem]{Corollary}

 \newtheorem{definition}[theorem]{Definition}
 
  \newtheorem{remark}[theorem]{Remark}

\newcommand{\Bohr}{\noindent \mbox{universal} }

\begin{document}
\title[Characterization of regular model sets]{Which Meyer sets are regular model sets? \\
A characterization via almost periodicity}

\author{Daniel Lenz}
\address{Mathematisches Institut, Friedrich Schiller Universit\"at Jena, 07743 Jena, Germany}
\email{daniel.lenz@uni-jena.de}
\urladdr{http://www.analysis-lenz.uni-jena.de}

\author{Christoph Richard}
\address{Department f\"{u}r Mathematik, Friedrich-Alexander-Universit\"{a}t Erlangen-N\"{u}rnberg,
Cauerstrasse 11, 91058 Erlangen, Germany}
\email{christoph.richard@fau.de}

\author{Nicolae Strungaru}
\address{Department of Mathematical Sciences, MacEwan University \\
10700 -- 104 Avenue, Edmonton, AB, T5J 4S2, Canada\\
and \\
Institute of Mathematics ``Simon Stoilow''\\
Bucharest, Romania}
\email{strungarun@macewan.ca}
\urladdr{https://sites.google.com/macewan.ca/nicolae-strungaru/home}

\dedicatory{We dedicate this work to Michael Baake and Franz Gähler\\
 on the occasion of their $65^{th}$ birthdays.}

\begin{abstract}
In 2012, Meyer introduced the notions of generalized almost periodic measure and almost periodic pattern  and proved that regular model sets in  Euclidean space  are almost periodic patterns. Here, we prove  the converse in a slightly more general setting. Specifically, we show that  a Meyer set in any $\sigma$-compact  locally compact abelian group is a regular model set if and only if it is an almost periodic pattern.
\end{abstract}

\maketitle


\section{Introduction}
This article is concerned with aperiodic order as exposed e.g.~in the article collections and monographs \cite{TAO,TAO2,BM,KLS,Moo}. The arguably most important  examples for aperiodic order are regular model sets. These sets were introduced and studied by Meyer in the seventies of the last century  in a harmonic analysis context \cite{me72}. Their r\^ole in aperiodic order was established by Moody~\cite{m97}, and fundamental geometric features were then shown by Lagarias \cite{L99}.

\smallskip

A central problem  in the field is to capture aperiodic order via notions of almost periodicity, see e.g.~the influential article by Lagarias \cite{L00}. Here we solve this problem for regular model sets, by showing that they can be characterized via generalized almost periodicity. This particular type of almost periodicity was brought forward by Meyer exactly to serve as an appropriate notion of almost periodicity for aperiodic order \cite{Mey2}. Loosely spoken, a discrete point set is generalized almost periodic if it can be approximated by almost periodic measures from above and from below, up to deviations of arbitrarily small mean. Meyer then showed that regular model sets are generalized almost periodic, and he calls such points sets  almost periodic patterns.
Our main result Theorem~\ref{thm:char-mod-set} shows that the converse is also true within the class of Meyer sets. It  also provides a slight strengthening of Meyer's original result in that it shows that almost periodic patterns agree with almost periodic measures up to sets of arbitrary small density. This formalizes a viewpoint sometimes expressed by physicists that regular model sets are almost periodic, up to a small error.

\smallskip

Our main result
is a consequence of a more general result, Theorem~\ref{thm:char-weighted-model}, which characterises generalized-almost-periodic measures with Meyer set support as being weighted model sets with Riemann integrable weight functions.
The proof of that theorem in turn relies on two ingredients:   We first re-analyse the construction of regular model sets based on cut-and-project schemes. In this context, we show existence of special cut-and-project schemes, called \textit{universal} below, which allow one to lift Bohr almost periodic functions on the real space to the internal space. As we show, they also allow one to lift strongly almost periodic measures supported within a given Meyer set.
The second ingredient is a certain approximation operator $\mathcal{T}$ acting on measures.  This operator smoothes a measure and restricts its support to a given Meyer set.
On a conceptual level, the introduction and investigation  of  $\Bohr$ cut-and-project schemes and the operator  $\mathcal{T}$ can be seen as the key achievements of this work.

\smallskip

As far as applications to aperiodic order are concerned, we could have restricted attention to Euclidean space. However we formulate our arguments in the slightly more general framework of $\sigma$-compact locally compact abelian groups. We feel that this makes the arguments more transparent. Moreover, let us  emphasize here that our arguments solely  rely on Bohr almost periodic functions, which makes locally compact abelian groups a natural setting.
We impose $\sigma$-compactness for the reader's convenience, as this allows one to compute densities of point sets along sequences of averaging sets that are called van Hove.
However, all our statements and their proofs below  hold for general locally compact abelian groups, with ``van Hove sequence'' being replaced by ``van Hove net'' \cite{PRS22}.

\smallskip

Our results complement  recent investigations \cite{LSS}
to characterize pure point diffraction (which is a key signature of a aperiodic order) via certain types of almost periodicity.  Specifically, in  \cite{LSS} the notions of Weyl almost periodic measure and Besicovitch almost periodic measure are put forward and the implications
\begin{eqnarray*}\mbox{generalized almost periodicity} &\Longrightarrow &\mbox{Weyl almost periodicity}\\
&\Longrightarrow & \mbox {Besicovitch almost periodity}
\end{eqnarray*}
are established. The main results of \cite{LSS} can be understood as saying that Besicovitch almost periodic measures are exactly the measures for which a convincing diffraction theory can be set up along a fixed van Hove sequence, and Weyl almost periodic measures are exactly the measures for which a convincing diffraction theory can be set up along any van Hove sequence. In this sense \cite{LSS} gives a complete picture for the r\^oles of Weyl- and Besicovitch almost periodicity in aperiodic order. The r\^ole of generalized almost periodicity in aperiodic order is now established by the main results of this article: it is exactly the form of almost periodicity present in the key class of examples.

\smallskip

Our article is organized as follows: The setting and basic notions are explained in Section \ref{sect-key}.  The core of our argument is presented in Section \ref{sect-core}, which features the notion of $\Bohr$ cut-and-project schemes and the operator $\cT$ mentioned above.  Our main results are then derived in Section \ref{sect-main}, and (counter)examples are discussed in Section~\ref{sec-ex}.
As we believe that the concept of $\Bohr$ cut-and-project scheme may be of independent interest, we include a further study of it in Section \ref{sec-Bohr}.

\section{Setting and key players}\label{sect-key}

We are concerned with  point sets  and almost periodicity  in a locally compact abelian group. Here we collect the necessary background for our considerations. The material of this section is well-known.  For early references see e.g.~\cite{me72, ARMA}. More references will be given along the way.

\subsection{Functions, measures, and sets}

Whenever $G$ is a locally compact abelian group, we denote a Haar measure on $G$ by $\theta_G$ or by $|\cdot|$. We denote integration  with respect to that Haar measure by $\int \dd s$ or by $\int  \dd t$.

\smallskip

The vector space of continuous functions from $G$ to $\CC$ is denoted by $C(G)$, the subspace of bounded continuous functions is denoted by $\Cb(G)$. The latter space is complete with respect to the supremum norm $\|\cdot\|_\infty$. The closed subspace of uniformly continuous and bounded functions is denoted by $\Cu(G)$,
the subspace of continuous and compactly supported functions is denoted by $\Cc(G)$. A real-valued function $f$ on $G$ is \textit{Riemann integrable} if for any $\varepsilon>0$ there exist $g,h\in \Cc (G)$ with $g\leq f\leq h$ and $\int_{G} (h(s)-g(s)) \dd s \leq \varepsilon$.
A complex-valued function is Riemann integrable if both its real part and its imaginary part are Riemann integrable.
We denote the space of complex-valued Riemann integrable functions on $G$ by $\cR(G)$.

\smallskip

Any linear functional on $\Cc(G)$ which is continuous in the inductive topology on $\Cc(G)$ can be identified with a complex Radon measure $\mu$ on $G$ via the Riesz representation theorem \cite{F99}. We will simply call such $\mu$ a \textit{measure}. A measure $\mu$ is  \textit{supported} within the closed set $C\subset G$ if $\mu (\varphi) =0$ for all $\varphi \in \Cc (G)$ with $\varphi =0$ on $G\setminus C$.
Consider $\varphi \in \Cc (G)$ and a measure $\mu$. We  mostly  write $\int \varphi \dd\mu  $  or $\int\varphi (t) \dd\mu(t)$ instead of $\mu(\varphi)$.
The \textit{convolution} $\varphi \ast \mu $ of $\varphi$ and $\mu$ is the function on $G$ given by
\[
(\varphi \ast \mu) (x) = \int_{G} \varphi (x-t) \dd\mu (t)
\]
for $x\in G$. A measure $\mu$ on $G$ is  \textit{translation bounded} if
 $\varphi*\mu$ is bounded for all $\varphi \in \Cc(G)$.
The vector space of all translation bounded measures on $G$ is denoted by $\M^{\infty}(G)$.

\smallskip

Any subset $\vL$ of $G$ meeting any compact set in only finitely many points  comes with a measure, called the \textit{Dirac comb} (of $\vL$), defined by
$$\delta_\vL =\sum_{x\in \vL} \delta_x \ .$$
Clearly, $\vL$ can be recovered from $\delta_{\vL}$. Hence, this construction allows one to think of such point sets also as measures.

\smallskip

 In this article, we are mainly interested in subsets of $G$ that are both uniformly discrete and relatively dense. Here, $\vL\subset G$ is called \textit{uniformly discrete} if there exists an open neighborhood $U$ of $0$ in $G$ such that $(s+U) \cap (t+ U) = \emptyset$ for all $s,t\in \vL$ with $t\neq s$. Moreover, $\vL\subset G$ is called \textit{relatively dense} if there exists a compact $K\subset G$ with $\vL + K = G$. We write $A+ B$ to denote the pointwise sum $\{a + b : a\in A,b\in B\}$  of two subsets $A,B$ of $G$.


\smallskip

\textbf{We now fix a $\sigma$-compact, locally compact abelian group $G$ for the remainder of the article.}

\subsection{Almost periodic measures and their mean}

We recall background on almost periodic functions and measures.
Let us denote by  $\tau_x$ the operator of translation by $x\in G$, which maps a function $f$ on $G$ to the function $\tau_x f$ given by $(\tau_x f)(t) = f(t-x)$ for $t\in G$.
A function $f\in \Cu (G)$ is \textit{Bohr almost periodic} if for every $\eps>0$ the  set  of its $\varepsilon$-almost periods
$$P_\eps=\{x\in G : \|f - \tau_x f \|_\infty <\eps\}$$
is relatively dense in $G$. The vector space of Bohr almost periodic functions on $G$ is denoted by $SAP(G)$. It is a function algebra with respect to pointwise multiplication that is closed in $(\Cu(G), \|\cdot\|_\infty)$.

\smallskip

There exists a unique compact group $G_\mathsf{b}$ admitting an injective group homomorphism $i : G\longrightarrow G_\mathsf{b}$ with dense range such that $SAP (G) = \{ f_\mathsf{b}\circ i : f_\mathsf{b}\in C(G_\mathsf{b})\}$.
It is called the \textit{Bohr compactification} of $G$. Define $M : SAP (G)\longrightarrow \CC$ by $M(f)=\theta_{G_\mathsf{b}}(f_\mathsf{b}\circ i)$, where $f=f_\mathsf{b}\circ i\in SAP(G)$ and where  $\theta_{G_\mathsf{b}}$ is the normalised Haar measure on $G_\mathsf{b}$. Then $M$ is positive linear and satisfies $M(1) =1$  and $M(\tau_x f) = M(f)$ for any $x\in G$ and $f\in SAP(G)$. In fact the latter properties uniquely determine $M$, as follows from the Riesz representation theorem for $G_\mathsf{b}$. The map  $M$ is called  the \textit{mean} on $SAP (G)$, and we call $M(f)$ the mean of $f$.

\smallskip

The mean of a Bohr almost periodic function can be computed by an averaging procedure: If $(A_n)$ is any  F\o lner sequence in $G$, we have
\[
M(f) = \lim_n \frac{1}{|A_n|} \int_{x+A_n} f(t) \dd t \ ,
\]
uniformly in $x \in G$. This is  a consequence of Eberlein's ergodic theorem \cite{EBE}.
Recall that a \textit{F\o lner sequence} in $G$ consists of compact subsets $(A_n)$ of positive Haar measure such that
$|A_n \triangle (x+A_n)|=o(|A_n|)$ as $n\to \infty$, for any $x\in G$.

\smallskip

A translation bounded measure $\mu\in \cM^\infty(G)$ is \textit{strongly almost periodic} if the function $\varphi \ast \mu$ is Bohr almost periodic for any $\varphi \in \Cc (G)$. The vector space of all strongly almost periodic measures is denoted by $\SAP (G)$.
The mean on $SAP(G)$ carries over to a map on $\SAP(G)$, again denoted by $M$ and called the \textit{mean} on $\SAP (G)$.  Indeed, the following holds as discussed in \cite[Lem.~4.10.6]{MoSt}. For the convenience of the reader we include the short proof.

\begin{lemma}[The mean on $\SAP (G)$]\label{lem:sapmean}
There exists a unique  linear map $M : \SAP(G)\longrightarrow \CC$ such that
\[
M(\varphi*\mu)= \theta_G(\varphi) \cdot M(\mu)
\]
for all $\mu \in \SAP(G)$ and $\varphi \in \Cc(G)$. The map $M$ satisfies the following properties:
\begin{itemize}
  \item[(a)] $M$ is positive.
  \item[(b)]  $M(\theta_G) =1$.
  \item[(c)]  $M(\mu\circ \tau_x )  = M(\mu)$ for any $x\in G$ and $\mu \in \SAP (G)$.
  \end{itemize}
\end{lemma}
\begin{proof}
Uniqueness is clear. As for existence we note that for positive $\mu$ the map
\[
\Cc (G)\longrightarrow \CC \ , \qquad \varphi \mapsto M(\varphi \ast \mu)\ ,
\]
is positive and invariant under translations. Hence it must be a multiple of the Haar measure. So we find a unique number $M_\mu$ with $M(\varphi\ast \mu) = \theta_G(\varphi) \cdot M_\mu$ for all $\varphi \in \Cc (G)$. Now, the statement follows easily for general $\mu$ by taking linear combinations.  The properties (a), (b), (c) are now clear.
\end{proof}

Similarly to $SAP (G)$, it is possible to compute the mean on $\SAP (G)$ by an averaging procedure. This procedure involves   special F\o lner sequences $(A_n)$ in $G$ introduced by Schlottmann \cite{Martin2} that are called  \textit{van Hove sequences}. Their characteristic feature is that
$$\lim_n \frac{\mu(\partial^K A_n)}{|A_n|} =0$$
holds  for any compact $K$ and any translation bounded measure $\mu$, where the \textit{$K$-boundary} $\partial^K A$ is defined by
$$
\partial^K A = ((A+K)\cap (G\setminus A^\circ)) \cup ((-K+\overline{(G\setminus A)}) \cap A) \ .
$$
Note that a compact set $A$ has a compact $K$-boundary.

\begin{remark} (a)
In Euclidean space, any sequence $(A_n)$ of compact balls of radius $n$ constitutes a van Hove sequence.

(b)  Existence of van Hove sequences for the $G$ treated in this article was shown in \cite{Martin2}. In fact,  any F{\o}lner sequence yields a van Hove sequence by a simple thickening procedure \cite{PRS22}.
\end{remark}

We have the following result \cite[Lem.~4.10.7]{MoSt}.

\begin{proposition}[Computing the mean via averaging]\label{prop:meanav}
Let $(A_n)$ be any van Hove sequence in $G$. We then have for all $\mu \in \SAP(G)$ that
\[
M(\mu) = \lim_n \frac{\mu(A_n)}{|A_n|} \ .
\]
Furthermore, this limit exists uniformly in translates. \qed
\end{proposition}

Having discussed  almost periodic measures  and their mean we can now define generalized almost periodicity. This notion  has been introduced by Meyer for translation bounded measures on Euclidean space \cite[Def.~2.31]{Mey2}. The extension to  $\sigma$-compact locally compact abelian groups has been discussed in \cite{LSS}.

\begin{definition}[Generalized almost periodicity and almost periodic patterns]\
\begin{itemize}
\item[(a)] A real measure $\varrho\in \mathcal M^\infty(G)$ is called {\rm generalized-almost-periodic (g-a-p)} if for every $\eps>0$ there exist $\mu_\eps,\nu_\eps\in\SAP(G)$ such that $\mu_\eps \leq \varrho \leq \nu_\eps$ and $M(\nu_\eps-\mu_\eps) < \eps$. A complex  measure $\varrho\in \mathcal M^\infty(G)$ is called    {\rm generalized-almost-periodic (g-a-p)} if both its real and imaginary parts are g-a-p. The  vector space of complex g-a-p measures is denoted by $\GAP (G)$.

\item[(b)] A  point set $\vL$ is called an {\rm almost periodic pattern} if its Dirac comb $\delta_{\vL}$ is a generalized-almost-periodic measure.
\end{itemize}
\end{definition}

\begin{remark}  (a) We note that $\varrho\in\mathcal M^\infty(G)$ is g-a-p if and only if it is a linear combination of positive g-a-p measures. (Observe that, if $\varrho$ is real and g-a-p, there exist strongly almost periodic $\mu, \nu$ such that $\mu \leq \varrho \leq \nu$. Then
\[
\varrho=\frac{1}{2} \left( (\varrho-\mu)-(\nu-\varrho)+(\nu-\mu) \right)
\]
is a linear combination of positive g-a-p measures.)

(b) It is not hard to see that the mean $M$ on $\SAP (G)$ can be extended to $\GAP(G)$. As we do not need this we refrain from further discussion.
\end{remark}

\subsection{Meyer sets, cut-and-project schemes, and the map $\vOmega$}
Here we discuss the  particular point sets most relevant to the study of aperiodic order.

\smallskip

\begin{definition}[Meyer set]
A uniformly discrete and relatively dense set $\vL$ in $G$ is called a {\rm Meyer set} if there exists a  finite  set $F$ in $G$ with $\vL - \vL \subset \vL + F$.
\end{definition}

\begin{remark} (a) Meyer sets were introduced and intensely studied by Meyer \cite{me72}. They were later called Meyer sets and placed in a central position in aperiodic order by Moody in his article  \cite{m97} whose terminology we adopt here.

(b) A relatively dense set $\vL$ is a Meyer set if and only if $\vL-\vL-\vL$ is uniformly discrete, see \cite[Ch.~II 14.2]{me72} and \cite[Lem.~5.7.1]{St17}. In Euclidean space, it has been shown by Lagarias that a relatively dense set $\vL$ is a Meyer set if and only if $\vL-\vL$ is uniformly discrete \cite{L99}. As shown by Lev and Olevskii, the latter condition can be weakened to $\vL-\vL$ having finite uniform upper density \cite{LO15}. Lagarias' characterisation extends to compactly generated locally compact abelian groups \cite{BLM}.
\end{remark}

\smallskip

Meyer gave a constructive description of all Meyer sets, which is based on the notion of cut-and-project scheme.

\begin{definition}[Cut-and-project scheme over $G$]
A {\rm cut-and-project scheme $(H,\cL)$} over $G$ consists of a locally compact abelian group $H$  and a discrete co-compact subgroup $\cL\subset G\times H$, called the {\rm lattice}, such that

\begin{itemize}
\item
the canonical projection
$\pi^G : \cL\longrightarrow G, (x,y)\mapsto x,$ is one-to-one;
\item  the canonical projection
 $\pi^H : \cL\longrightarrow H, (x,y)\mapsto y,$
has dense range.
\end{itemize}
\end{definition}
When working with a cut-and-project scheme $S=(H,\cL)$ over $G$, we will often abbreviate $L=\pi^G(\cL)$. Any cut-and-project scheme comes with a canonical  map on $L$ called   \textit{star map} given by   $$\star: L\to H, \qquad x\mapsto x^\star,$$
 where  for $x\in L$ the point
$x^\star \in H$ is the unique element of $H$ such that $(x,x^\star)\in \cL$. (Here, uniqueness follows from injectivity of $\pi^G$ on $\cL$.) To any  $W\subset H$ we can then associate the \textit{cut-and-project set}
\[
\oplam_{S}(W)=\{x\in L: x^\star\in W\} \ .
\]
Whenever the cut-and-project scheme $S$ is clear from context, we will simply write $\oplam(W)$ instead of $\oplam_{S}(W)$. By standard theory, $\oplam (W)$ is uniformly discrete if $\overline{W}$ (the closure of $W$)  is compact, and $\oplam(W)$ is relatively dense if $W^\circ$ (the interior of $W$)  is not empty. We call $W$ the \textit{window} underlying $\oplam_{S}(W)$. A point set $\vL$ in $G$ is called a \textit{model set} if there exist a cut-and-project scheme $S=(H,\cL)$ over $G$ and $W\subset H$  with non-empty interior and compact closure such that $\vL = \oplam(W)$ holds.   For us the following strengthening of the concept of model set will be crucial.

\begin{definition}[Regular model set]
A point set $\vL$ in $G$ is a {\rm regular model set} if there exist a cut-and-project scheme $(H,\cL)$ over $G$ and   $W\subset H$  with non-empty interior and compact closure whose
topological boundary $\partial W$ has zero Haar measure
such that $\vL = \oplam(W)$ holds. Such $W$ is then called a regular window.
\end{definition}

The following result readily derives from Meyer's work on harmonious sets \cite{me72}. A proof based on almost periodicity is given in \cite{St17}. A related result about constructing a cut-and-project scheme out of suitable diffraction data appears in \cite{BM04}.

\begin{theorem}[Meyer's characterization] \label{thm:meychar} Let $\vL$ be a relatively dense set in $G$. Then the following assertions  are equivalent.
\begin{itemize}
\item[(i)] $\vL$ is a Meyer set.
\item[(ii)] $\vL$ is a subset of a model set.
\item[(iii)] $\vL$ is a subset of a regular model set.
\end{itemize}
\qed
\end{theorem}
The construction of point sets out of cut-and-project sets via $\oplam$ can be generalized to give -- so to speak --  weighted point sets. To make this precise we think of point sets as measures by invoking Dirac combs.  Specifically, if  $S=(H,\cL)$ is some cut-and-project scheme over $G$, then any point set $\oplam (W)$ in $G$ with a window $W$ having compact closure gives rise to the translation bounded measure
$$\delta_{\oplam(W)}= \sum_{x\in L} 1_W (x^\star) \delta_x \ .$$
This is a special case of the map
$\vOmega=\vOmega_S$ from bounded functions on $H$ vanishing outside some compact set to the translation bounded  measures  on $G$ defined by
$$\vOmega(h) =\sum_{x\in L}  h(x^\star) \delta_x \ .$$
Indeed,  for $h: H\to \CC$ bounded by a finite number $C$ and  vanishing outside of a compact $K\subset H$, the measure $\vOmega(h)$ is a point measure vanishing outside the uniformly discrete set  $\oplam (K)$, and the mass of a single point of $\vOmega(h)$ is bounded by $C$, whence  $\vOmega(h)$ is translation bounded. For relatively compact $W\subset H$ we then find $\delta_{\oplam(W)} = \vOmega(1_W)$.

\smallskip

This construction produces many strongly almost periodic measures \cite[Lem.~7.3]{LR}.

\begin{proposition}[Strong almost periodicity of $\vOmega (h)$]\label{omega-is-sap}
Let $S = (H,\cL)$ be a cut-and-project scheme over $G$ with associated map $\vOmega$. Then  $\vOmega(h)\in\SAP(G)$ for any $h\in \Cc (H)$. \qed
\end{proposition}

The mean of such a strongly almost periodic measure can be computed using a similar approach as in the proof of Lemma~\ref{lem:sapmean}. This gives the following result, see  \cite[Thm.~9.1]{LR} for a proof.

\begin{lemma}[Computing the mean of $\vOmega(h)$]\label{computing-the-mean}
Let $S = (H,\cL)$ be a cut-and-project scheme over $G$ with associated map $\vOmega$. Then
there exists a constant $D_S>0$ such that
\begin{equation}\label{eq:dfvO}
M(\vOmega (h)) = D_S  \cdot\int_{H} h  \dd \theta_H
\end{equation}
for all $h\in \Cc (H)$.
\qed
\end{lemma}

\begin{remark}[Density formulae] The positive constant $D_S$ in the  lemma  can be understood as the density of the lattice $\cL$ in $G\times H$ and the above equation is sometimes called the density formula for $\vOmega(h)$.  Density formulae for model sets date back to the early works of Meyer and have been re-derived and extended using a variety of methods such as harmonic analysis, dynamical systems and almost periodicity. A discussion of the literature and a recent geometric proof can be found in \cite{PRS22}.
\end{remark}

The previous lemma can be used to infer generalized almost periodicity of $\vOmega(h)$ for Riemann integrable weight functions $h$:

\begin{proposition}[Generalized almost periodicity of $\vOmega(h)$]\label{prop:gapvO}
Let $S = (H,\cL)$ be a cut-and-project scheme over $G$ with associated map $\vOmega$. Then  $\vOmega(h)\in\GAP(G)$ for any $h\in R(H)$.
\end{proposition}

\begin{proof}
Without loss of generality we assume that $h$ is real-valued.
Let $\varepsilon>0$ be arbitrary. Since $h:H\to \mathbb R$ is Riemann integrable, there exist $f, g \in \Cc(H)$ such that $f \leq h \leq g$ and $\int_{H}(g-f) \, \dd \theta_{H} \leq \varepsilon$.
Define $\mu=\vOmega(f)$  and  $ \nu= \vOmega(g)$.
Then, $\nu$ and $\mu$ belong to $\SAP (G)$ by Proposition \ref{omega-is-sap}. Moreover, by construction we have
$\mu\leq \varrho \leq \nu$. From Lemma~\ref{computing-the-mean}
we now find
\begin{displaymath}
  M(\nu-\mu) = D_{S} \cdot \int_H (g-f) \, {\rm d}\theta_{H}  \leq D_{S} \cdot \varepsilon \ .
\end{displaymath}
As $\varepsilon>0$ was arbitrary, this shows  that $\vOmega(h)$ is a g-a-p measure.
\end{proof}

\begin{remark} It is not hard to see that  the density formula Eqn.~\eqref{eq:dfvO}  continues to hold for the mean of the g-a-p measure $\vOmega(h)$ for Riemann integrable $h$. Moreover, the mean exists uniformly in translates.  This extends the particular case $h=1_W$ being the characteristic function of a regular window $W$. This case  is treated in \cite[Cor.~4.3]{Mey2} and motivated the definition of generalized almost periodicity.
\end{remark}

\section{The core of the argument}\label{sect-core}

In this section we present the core of our argument to characterize  g-a-p measures of Meyer set support.
The two main ingredients are  a certain cut-and-project scheme, which we call $\Bohr$ cut-and-project scheme, and a certain operator on measures, denoted $\cT$ below.

\begin{definition}[The $\Bohr$ cut-and-project scheme]\label{def:bc}
A cut-and-project scheme $(H,\cL)$ over $G$ is called {\rm universal}   if
there exists a map
$$
{}^\sharp: SAP(G)\to \Cb(H) \text{ with } f(x)=f^\sharp(x^\star)  \text{ for  all } x\in L \ .
$$
\end{definition}

\smallskip

\begin{remark}
A $\Bohr$  cut-and-project scheme satisfies a certain uniqueness property and a certain factorisation property. Moreover, the  function $f^\sharp$ belongs even to $SAP(H)$ for any $f\in SAP(G)$. These features are not relevant for the proof of our main result. Thus, we only prove  them later, together with various further characterizations in Section~\ref{sec-Bohr}.
\end{remark}

We note that any $\Bohr$ cut-and-project scheme has an injective star map,
since the Bohr almost periodic functions on $G$ separate the points of $G$. This property will be important below.

 For us the following feature of a $\Bohr$ cut-and-project scheme $(H,\cL)$ will be most relevant: For any $h\in \Cc (H)$ and any $f\in SAP (G)$ we have
$$f\cdot \vOmega(h) = \vOmega(f^\sharp \cdot h) \ .$$
Indeed, both sides are point measures with point masses on the points $x\in L$ given by
$f(x)\cdot h(x^\star) = f^\sharp (x^\star) \cdot h(x^\star)$. Whenever $(H,\cL)$ is a $\Bohr$ cut-and-project scheme and $h\in \Cc(H)$ and $\varphi \in \Cc(G)$ are given, we can then define the operator
$$\cT: \cM^\infty(G) \to \cM^\infty(G) \ ,  \qquad \cT (\mu) =(\varphi\ast \mu) \cdot \vOmega (h) \ .$$
Then,  for any strongly almost periodic measure $\mu$ the function  $\varphi\ast \mu$ is Bohr almost periodic and we obtain
$$\cT (\mu) = (\varphi\ast \mu) \cdot  \vOmega(h) = \vOmega ((\varphi\ast \mu)^\sharp \cdot h) \ .$$
Hence, $\cT$ maps  $\SAP(G)$ into $\vOmega (\Cc(H))$. Moreover,  for suitable $\varphi$ and $h$   the operator $\cT$ will fix all measures supported on a given Meyer set.  These properties of $\cT$ will be a key to our reasoning below.

\smallskip

The next proposition settles the existence of a $\Bohr$ cut-and-project scheme.  For a fixed cut-and-project scheme $S$ over $G$, we say that a subset of $G$ is \textit{associated to $S$} if it  is contained in a model set from $S$.  Clearly, a subset $\vL$ of $G$ is  associated to $S$ if and only if there exists an $h\in \Cc (H)$ with $\delta_\vL\leq \varOmega (h)$.
Denote by $\cM_S(G)\subset \cM^\infty(G)$ the collection of translation bounded measures whose support is associated to $S$.

We will show  that to any cut-and-project scheme there exists a $\Bohr$  cut-and-project scheme with the same associated Meyer sets. To ease the formulation we call two cut-and-project schemes over $G$ \textit{compatible} if they have the same associated sets. Equivalently, two cut-and-project schemes $S$ and $S'$ over $G$ are compatible if  $\cM_S(G)=\cM_{S'}(G)$ holds.

\begin{proposition}[Existence of a compatible  $\Bohr$ cut-and-project scheme]\label{prop:excBc}
To any cut-and-project scheme there exists a compatible $\Bohr$ cut-and-project scheme. In particular, for any Meyer set there exists a $\Bohr$ cut-and-project scheme to which it is associated.
\end{proposition}

\begin{proof}
Let $S=(H, \cL)$ be any cut-and-project scheme over $G$ with $L= \pi^G(\cL)$. We define a cut-and-project scheme $S_\mathsf{u}=(H_\mathsf{u}, \cL_\mathsf{u})$ using the Bohr compactification $G_\mathsf{b}$ of $G$ as follows:
\begin{displaymath}
H_\mathsf{u}=\overline{\{ (x^\star, i(x)) : x \in L \}} \subset H \times G_\mathsf{b} \ ,
\end{displaymath}
where $i: G \to G_\mathsf{b}$ is the canonical injection map.
We define $\cL_\mathsf{u} \subset G\times H_\mathsf{u}$ by $\cL_\mathsf{u} = \{ (x, x^\star, i(x) ) : x \in L \}$.  Then,
$S_\mathsf{u}=(H_\mathsf{u}, \cL_\mathsf{u})$ is a cut-and-project scheme over $G$ (compare
 \cite[Lem.~4]{LLRSS} as well).

\smallskip

\noindent \textit{$S$ and $S_\mathsf{u}$ are compatible:}
Let $\vL$ be any set associated to $S$. Then, from $\delta_\varLambda \leq \vOmega_S (h)$ for some $h\in \Cc(H)$ we easily find $\delta_\varLambda \leq \vOmega_{S_\mathsf{u}} (h\otimes 1)$.  Conversely, for any set $\vL$ associated to $S_\mathsf{u}$ there exists an $h_\mathsf{u} \in \Cc (H_\mathsf{u})$ with $\delta_{\vL}\leq \vOmega_{S_\mathsf{u}} (h_\mathsf{u})$. Then, there exists a compact set $K$ in $H$ such that $h_\mathsf{u}$ is supported in $(K\times G_\mathsf{b})\cap H_\mathsf{u}$. It follows that $\delta_{\vL}\leq \vOmega_S (h)$ for any $h\in \Cc (H)$ with $h= 1$ on $K$.

\smallskip

\noindent \textit{$S_\mathsf{u}$ is a $\Bohr$ cut-and-project scheme}:  Clearly, the map
$$()^\sharp : SAP(G) \longrightarrow \Cu(H_\mathsf{u}) \ , \quad f\mapsto (1\otimes f_\mathsf{b})|_{H_\mathsf{u}}$$
 has the desired properties. Here $|_{H_\mathsf{u}}$ denotes restriction to $H_\mathsf {u}$.

\smallskip

\noindent The ``in particular'' statement now follows with Theorem~\ref{thm:meychar}.
\end{proof}

\begin{remark}
The cut-and-project scheme appearing in the preceding proof  has already been used in the different context of deformed model sets \cite{LLRSS}.
\end{remark}

We next come to the  operator $\cT$  acting on measures. It produces a measure supported within a given fixed Meyer set via smoothing.

\begin{proposition}[The operator $\cT$]
\label{prop:sap-meyer} Let $\vL$ be a Meyer set in $G$. Take $\omega\in \SAP(G)$ and a Meyer set $\varGamma$ in $G$ be such that $\delta_\vL\leq \omega \leq \delta_{\varGamma}$. Fix one $\psi\in \Cc(G)$ with $0\leq \psi \leq 1$ and $\psi (0) =1$, such that any translate of $\supp(\psi)$ meets $\vG$ in at most one point.
Define the operator $\cT=\cT_{\psi,\omega}: \cM^\infty(G) \to \cM^\infty(G)$ via
\[
\cT(\mu)= (\psi \ast \mu)\cdot \omega \,.
\]
Then the following hold.
\begin{itemize}
\item[(a)] $\cT$ is linear and positive.
\item[(b)] For all $\mu \in \cM^\infty(G)$ we have $\supp(\cT(\mu)) \subset \varGamma$.
\item[(c)] $\cT(\mu)=\mu$ for all $\mu \in \cM^\infty(G)$ supported within $\vL$.
\item[(d)] $\cT(\SAP(G)) \subset \SAP(G)$.
\item[(e)] For any translation bounded non-negative measure $\mu$ and any compact $A$ the inequality $\cT(\mu)(A)\leq \mu (A-K)$ holds, where  $K=\supp(\psi)$.
\item[(f)] For  all non-negative $\mu \in \SAP(G)$ we have $M(\cT(\mu)) \leq M(\mu)$.
\end{itemize}
\end{proposition}

\begin{remark} It is not hard to see that to any Meyer set $\vL$ a measure $\omega$ with the claimed properties exists.
An explicit construction will be given in the proof of our next lemma.
\end{remark}

\begin{proof}

\noindent (a) and (b) are obvious.

\noindent (c) Consider arbitrary $\mu \in \cM^\infty(G)$ supported within $\vL$. Note first $\supp(\cT(\mu)) \subset \vG$ as $\supp(\omega) \subset \vG$.
But for any $x \in \vG$ we have
\begin{displaymath}
\left((\psi*\mu) \cdot \omega\right)(\{ x \})= (\psi*\mu)(x)\cdot  \omega(\{ x \})= \omega(\{ x \}) \sum_{y \in \vL} \psi(x-y) \mu(\{y\})  \,.
\end{displaymath}
For $y \in \vG$ we have $\psi(x-y)=0$ if $y\ne x$ by choice of $\psi$, which shows
\begin{displaymath}
\left((\psi*\mu) \cdot \omega\right)(\{ x \})
= \omega(\{ x \}) \psi(0) \mu(\{x\}) = \omega(\{ x \}) \mu(\{x\}) = \mu(\{x\}) \ .
\end{displaymath}
Here the last equality holds as $\mu(\{x \}) \neq 0$ implies $x \in \vL$ and $\omega(\{x\})=1$. We thus have $\cT(\mu)=\mu$.

\noindent (d) We have $\omega \in \SAP(G)$ and $\psi*\mu \in SAP(G)$ as $\mu\in\SAP(G)$.  Then $\cT(\mu) \in \SAP(G)$ by \cite[Thm.~6.1]{ARMA}.

\noindent (e)
For any compact $A\subset G$, we use non-negativity of $\mu,\psi, \omega$ and the inequality $\omega \leq \delta_{\Gamma}$ in order to estimate
\begin{align*}
\cT(\mu)(A) & =\sum_{x \in \vG  \cap A} (\psi*\mu)(x) \cdot \omega(\{ x \}) \leq
\sum_{x \in \vG \cap A} (\psi*\mu)(x) \\
&= \sum_{x \in \vG \cap A} \int_{G} \psi(x-y) \dd \mu(y)
= \int_{G} \sum_{x \in \vG \cap A}  \psi(x-y)  \dd \mu(y) \ .
\end{align*}
Note that $x \in A$ and $ \psi(x-y)  \neq 0$ imply $y \in A-K$. We thus have
\begin{align*}
\cT(\mu)(A) & =\int_{G} \sum_{x \in \vG \cap A}  \psi(x-y)  \dd \mu(y) =  \int_{G}   \sum_{x \in \vG \cap A} 1_{A-K}(y) \psi(x-y) \dd \mu(y) \\
&= \int_{G} 1_{A-K}(y)  \sum_{x \in \vG \cap A} \psi(x-y) \dd \mu(y)\leq \int_{G} 1_{A-K}(y) \left( \sum_{x \in \vG }  \psi(x-y) \right) \dd \mu(y) \\
&\le  \int_{G} 1_{A-K}(y)  \dd \mu(y)  = \mu(A-K) \ .
\end{align*}

\noindent (f)
Let  $(A_n)$ be a van Hove sequence in $G$. Due to the elementary estimate $A-K\subset A + \partial^{-K}A$ we then have for all $\mu \in \SAP(G)$
\begin{align*}
M(\cT(\mu))&= \lim_{n} \frac{1}{|A_n|} \cT(\mu)(A_n) \leq \limsup_{n} \frac{1}{|A_n|}  \mu (A_n-K) \\
&\leq \limsup_{n}\frac{1}{|A_n|} \left( \mu (A_n)+\mu(\partial^{-K}A_n)\right) =\lim_{n}\frac{1}{|A_n|}   \mu (A_n) = M(\mu)  \ .
\end{align*}
This finishes the proof.
\end{proof}

Based on our study of $\cT$ and the concept of $\Bohr$ cut-and-project scheme, we can  now state and prove the  following approximation lemma. It  will be a key argument in the proof of Theorem~\ref{thm:char-weighted-model}.

\begin{lemma}[Approximation in a compatible $\Bohr$ cut-and-project scheme]\label{approx}

Let $\vL$ be a Meyer set in $G$ and let $S=(H,\cL)$ be a $\Bohr$ cut-and-project scheme over $G$ such that $\vL$ is associated to $S$. Then, for any $\varrho\in \cM^\infty(G)$ supported within $\vL$ the following hold.
\begin{itemize}
\item[(a)] Assume  $\varrho \leq \nu$ for some  $\nu \in \SAP(G)$. Then  there exists $g \in \Cc(H)$  such that
$\varrho \leq \vOmega(g)$ and $M(\vOmega(g))\leq M(\nu)$.
\item[(b)] Assume  $\varrho\ge \mu$ for some $\mu\in \SAP(G)$. Then there exists $f \in \Cc(H)$  such that $\varrho\ge \vOmega(f)$ and $M(\vOmega(f)) \geq M(\mu)$.
\end{itemize}
\end{lemma}
\begin{proof}
(a) As $\vL$ is a Meyer set, we find $W\subset H$ with non-empty interior and compact closure such that $\delta_\vL\leq \vOmega (1_W)$. Let $V\subset H$ be an open set with compact closure such that $\overline W\subset V$. By Urysohn's lemma \cite[Thm.~2.12]{Rud}, we find $h\in \Cc(H)$ such that $1_{\overline{W}}\le h \le 1_V$. Then $\omega=\vOmega(h)\in \SAP(G)$ satisfies $\delta_\vL\le \omega\le \delta_\varGamma$ for the Meyer set $\varGamma=\oplam(V)$.
Choose $\psi\in \Cc (G)$ with $\psi\geq 0$ and $\psi (0) =1$ and such that any translate of the support of $\psi$ meets $\varGamma$ in at most one point.
Then the operator defined by $\cT (\mu)= (\psi \ast \mu) \cdot \omega$ for $\mu\in\cM^\infty(G)$
satisfies the statements of Proposition \ref{prop:sap-meyer}.
In particular $\cT (\varrho) = \varrho$
holds, as $\varrho$ is supported within $\vL$. As $\nu$ is strongly almost periodic, the function $f=\psi \ast \nu$ is Bohr almost periodic. As $S$ is a universal cut-and-project scheme, we find
$$\cT (\nu) = (\psi \ast \nu) \cdot \vOmega (h) = f \cdot\vOmega (h)= \vOmega (f^\sharp \cdot h)= \vOmega(g),$$
where  $g=f^\sharp \cdot  h\in \Cc (H)$. Let now $(A_n)$ be any van Hove sequence.
With $K = \supp(\psi)$ we then obtain by Proposition~\ref{prop:sap-meyer} (e) for each $n$ the estimate
\begin{align*}
\left( \vOmega(g)-\varrho \right) (A_n) &= \left(\cT (\nu) -\cT(\varrho)\right)(A_n)=  \cT (\nu - \varrho)(A_n) \\
&\leq \nu(A_n-K) - \varrho(A_n-K) \ .
\end{align*}
It follows that
\[
\frac{1}{|A_n|}\vOmega(g)(A_n) \leq \frac{1}{|A_n|}\nu(A_n-K)+\frac{1}{|A_n|} \left( \varrho(A_n) -\varrho(A_n -K) \right)  \,.
\]
The van Hove condition then gives $M(\vOmega(g))\leq M(\nu)$. Moreover, the properties of $\cT$ give
\begin{displaymath}
\varrho=\cT(\varrho)\le \cT(\nu)=\vOmega(g) \ .
\end{displaymath}
\noindent (b) This follows from (a) by passing from $\varrho$ to $-\varrho$ and from $\mu$ to $-\mu$.
\end{proof}

To give  an impression of the power of the approximation lemma, we use it to derive the following result on representing strongly almost periodic measures via $\vOmega$.  This result is of independent interest, but is not needed in the proof of Theorem~\ref{thm:char-weighted-model}.  However it provides  some perspective as  the latter theorem can be seen as  an extension of it, with  $\Cc(H_\mathsf{u})$  replaced by  Riemann integrable functions $\cR(H_\mathsf{u})$, and with strongly almost periodic measures replaced by generalized almost periodic measures.

\begin{corollary}[Representing $\SAP(G)$ via $\Bohr$ cut-and-project schemes] \label{cor-rep}
Let $S$ be a cut-and-project scheme over $G$ and let $S_\mathsf{u}=(H_\mathsf{u},\cL_\mathsf{u})$ be any compatible $\Bohr$ cut-and-project scheme.
Then, the map
$$\Cc (H_{\mathsf{u}})\longrightarrow  \SAP (G) \cap \cM_S (G),  \quad h\mapsto \vOmega_{S_\mathsf{u}}(h),$$
is a bijection. \qed
\end{corollary}

\begin{remark}[Characterization of universal cut-and-project scheme]
Note that the above corollary applies in particular to $S=S_{\mathsf{u}}$, as any universal cut-and-project scheme is compatible with itself. In fact, the assertion of the corollary for $S=S_{\mathsf{u}}$ characterises the universal cut-and-project scheme compatible with $S$ up to isomorphism. This follows from Theorem~\ref{thm:char-bu} below, in conjunction with Proposition~\ref{prop:uni}.
\end{remark}

\begin{proof}[Proof of Corollary~\ref{cor-rep}]
As $\cM_S(G)=\cM_{S_\mathsf{u}}(G)$, we assume without loss of generality $S=S_{\mathsf{u}}$ in the following.

The map $\vOmega$ is clearly injective (as the range of the star map is dense in $H$). We now show that it is surjective. Let $\varrho$ be an almost periodic measure supported on the Meyer set $\vL$, which is associated to $S$. Assume without loss of generality that $\varrho$ is real. Then, from the approximation lemma with $\mu = \varrho = \nu $ we infer that there exist $f,g\in \Cc (H)$ satisfying $\vOmega(f)\leq \varrho\leq \vOmega(g)$ with $M(\vOmega(g))\leq M(\varrho)$ and $M(\varrho)\leq M(\vOmega(f))$. Using monotonicity of the mean, we infer from these inequalities  $M(\vOmega(f))=M(\vOmega(g))$.
Observe also that the first inequality implies $f\leq g$ by continuity, as the range of the star map is dense in $H$.
Now invoking  Lemma \ref{computing-the-mean} we find
$$\int_{H}(g-f) \dd\theta_{H} = 0 \ .$$
From $f\leq g$ and continuity of $f$ and $g$  we conclude $f =g$, and $\vOmega(g)=\varrho$ follows.
\end{proof}

\section{G-a-p measures supported within Meyer sets}\label{sect-main}

Here we characterize regular model sets as almost periodic patterns. This characterization will be a  consequence of a more general result on representing  g-a-p measures  whose support is contained in some Meyer set. Recall that $\cM_S(G)$ denotes the collection of translation bounded measures whose support is associated to the cut-and-project scheme $S$ over $G$.

\begin{theorem}[Representing $\GAP(G)$ via universal cut-and-project schemes]\label{thm:char-weighted-model}
Let $S$ be a cut-and-project scheme over $G$ and let  $S_\mathsf{u}=(H_\mathsf{u},\cL_\mathsf{u})$ be a compatible $\Bohr$ cut-and-project scheme. Then, the map
$$ \cR(H_{\mathsf{u}})\longrightarrow \GAP(G)\cap \cM_S (G), \quad
h\mapsto \vOmega_{S_\mathsf{u}}(h),$$
is onto.
\end{theorem}

\begin{remark}[The map $\vOmega$]
The above map $h\mapsto \vOmega_{S_\mathsf{u}}(h)$ has a few additional properties: It is linear and its kernel consists exactly of the Riemann integrable functions vanishing on $L^{\star_\mathsf{u}}\subset H_\mathsf{u}$.  Moreover, it is not hard to see that it is an isometry when $R(H_\mathsf{u})$ is equipped with the seminorm $\|h\|_{R} = D_{S_\mathsf{u}}\cdot \theta_{H_\mathsf{u}}(|h|)$ and $\GAP (G)$ is equipped with the seminorm $\|\mu\|_{\GAP} =M(|\mu|)$, where $|\mu|$ denotes the total variation of the measure $\mu$.
\end{remark}

\begin{proof}[Proof of Theorem~\ref{thm:char-weighted-model}]
As $\cM_S(G)=\cM_{S_\mathsf{u}}(G)$, we assume without loss of generality $S=S_{\mathsf{u}}$ in the following.

Note that $\vOmega(h)$ is indeed a g-a-p  measure by Proposition~\ref{prop:gapvO}. As $h$ is Riemann integrable, we can take a compact set $W\subset H$ having non-empty interior such that $\supp(h)\subset W$. Then $\supp(\Omega(h))\subset \oplam(W)$, which means that $\vOmega(h)$ has support associated to $S$.
%
%
We show that the map $\vOmega$ is onto: Let $\varrho$ be an arbitrary  g-a-p measure with support associated to $S$. Without loss of generality we assume that $\varrho$ is a real measure.
For every $n\in \mathbb N$ we can find  $\mu_n, \nu_n \in \SAP(G)$ with
$$\mu_n \leq \varrho \leq \nu_n \mbox{ and }  M(\nu_n-\mu_n) \stackrel{n\to \infty }{\longrightarrow} 0\ .$$
By Lemma~\ref{approx}, we can choose $f_n, g_n \in \Cc(H)$ with
 $$\varrho \geq \vOmega(f_n)  \mbox{ and } M(\vOmega(f_n)) \geq M(\mu_n) \ ,$$
 and
 $$\varrho \leq \vOmega(g_n) \mbox{ and }  M(\vOmega(g_n)) \leq M(\nu_n) \ ,$$
respectively for each $n\in\NN$.
In particular we have for arbitrary $m,n\in \NN$ that
$$f_m(x^\star)\le \varrho(\{x\})\le g_n(x^\star)$$ for $x\in L$. This implies $f_m\le g_n$ for each $n,m\in\NN$ by denseness of $L^\star$ in $H$.  We thus have $\sup_{n} f_n \le \inf_n g_n$. We now  define  $h: H \to \RR$ by
\[
h(y)=
\left\{
\begin{array}{lc}
\varrho(\{\star^{-1}(y)\}) & y\in L^\star \\
\sup_n f_n (y) & y \notin L^\star
\end{array}
\right. \ ,
\]
where we used injectivity of the star map. Then by construction $\vOmega(h)=\varrho$ and $f_n\le h \le g_n$.
Moreover we have
\begin{displaymath}
0\leq M(\vOmega(g_n))-M(\vOmega(f_n)) \leq M(\nu_n)-M(\mu_n) \stackrel{n\to \infty }{\longrightarrow} 0\ .
\end{displaymath}
Using the density formula Lemma~\ref{computing-the-mean},
we conclude that $h$ is Riemann integrable.
\end{proof}

With   Theorem~\ref{thm:char-weighted-model} we find the following  characterization of g-a-p measures supported within Meyer sets.
\begin{corollary}[Characterization of g-a-p measures supported within Meyer sets]
Let $\varrho\in \cM^\infty(G)$ be a discrete measure supported within some Meyer set. Then $\varrho\in\GAP(G)$  if and only if there exists a cut-and-project scheme $S=(H,\cL)$ over $G$ such that $\varrho=\vOmega_S(h)$ for some $h\in R(H)$.
\end{corollary}

\begin{proof}
\noindent \emph{``if part''}: This is the statement of Proposition~\ref{prop:gapvO}.


\smallskip

\noindent \emph{``only if part''}: Assume that $\varrho\in\GAP(G)$ is supported within some Meyer set $\vL$. By Proposition~\ref{prop:excBc}, we might take a universal cut-and-project scheme $S=(H,\cL)$ to which $\vL$ is associated. Now the claim follows from Theorem~\ref{thm:char-weighted-model}, as any universal cut-and-project scheme is compatible with itself.
\end{proof}

We now specialize to Meyer sets.
 For our reasoning we need the following auxiliary statement. It shows that a Riemann integrable function that attains its values in $\{0,1\}$ on a dense set must essentially  be the characteristic function of a set. This is certainly known. As we have not found a reference, we include a proof for completeness.
\begin{lemma}\label{lem:2_WRI}
Let $H$ be a locally compact abelian group and let  $h\in \cR(H)$. Assume that there exists a dense set $D \subset H$ such that $h(D) \subset \{0,1 \}$.
Then there exists a set $W \subset H$ with compact closure  such that $1_{W}\in \cR(H)$ and
$h(y)=1_{W}(y)$ holds for all $y\in D$.  In particular, the topological boundary of $W$ has Haar measure $0$.
\end{lemma}

\begin{proof}
Assume without loss of generality that $h$ is real-valued. For every $n\in\NN$ pick $f_n, g_n \in \Cc(H)$ such that $f_n \leq h \leq g_n$ and
$\int_H (g_n-f_n) \, {\rm d}\theta_H \to 0$ as $n\to \infty$. Note  $f_n\le 1$ and $g_n\ge0$ on $H$ by continuity and define the set $W\subset H$ by
\[
W= \bigcup_{m\in \NN}\{ f_m > 0 \} \cup \{h=1\}\ .
\]
 We claim $h(y)=1_{W}(y)$ for all $y \in D$:  Recall that $h(y) \in  \{0,1 \}$ for all $y \in D$ by assumption. If $h(y)=1$, then $y \in W$ and hence $1_W(y)=1$. If $h(y)=0$, then $f_n(y) \leq0$ for all $n\in \NN$. Hence $y \notin W$ in this case, which means $1_{W}(y)=0$. Thus $h(y)=1_{W}(y)$ for all $y \in D$.

\smallskip

\noindent Moreover we claim  $f_n \leq 1_W \leq g_n$ for all $n\in \NN$, which implies $1_W\in R(H)$:

\smallskip

\noindent We first show $f_n \leq 1_W$. As  $f_n\le 1$, we only have to consider $y \notin W$. Then $f_n(y) \leq 0 = 1_{W}(y)$, where the first inequality follows from the definition of $W$.

\smallskip

\noindent We next show $1_W \leq g_n$. Since $0\le g_n$, we only need to show that $g_n(y) \geq 1$ for all $y\in W$. Let thus $y\in W$. If $h(y)=1$, then $1_{W}(y)=1=h(y) \leq g_n(y)$, and the claim follows. On the other hand, if $f_m(y) > 0$ for some $m\in \NN$, then $0<f_m\le h \le g_n$ on some neighborhood $V$ of $y$ by continuity.
By assumption on $D$ and continuity of $g_n$ we thus get $1\le g_n$ on $\overline{V\cap D}$. Noting that $y\in V\subset \overline{V\cap D}$, we infer $1_W(y)=1\le g_n(y)$ as claimed.

\smallskip

\noindent We finally show $\theta_H(\partial W)=0$:  As $1_W$ is Riemann integrable, we find for arbitrary $\varepsilon >0$ functions $f,g \in \Cc (H)$ with $f\leq 1_W \leq g$ and $\theta_H(g-f)<\varepsilon$. As continuity implies $f\le 1_{W^\circ}$ and $1_{\overline{W}}\le g$, we obtain $\theta_{H} (\partial W) \le \theta_H(g-f) <\varepsilon$. Thus $\theta_H(\partial W)=0$.
\end{proof}

Our characterization of regular model sets is now obtained by specializing Theorem \ref{thm:char-weighted-model} to point sets and invoking the preceding lemma. We will also show a third equivalent condition in this case, which does not seem to extend to arbitrary g-a-p measures with Meyer set support.

\smallskip

We need the following pieces of notation. For two point measures $\mu$ and $\nu$ on $G$  we denote
\[
\{\mu \neq \nu\}= \{ x\in G : \mu(\{x\}) \neq \nu(\{x\}) \} \ .
\]
Moreover, we fix a van Hove sequence $(A_n)$ and define
for  a uniformly discrete subset $\vL$ of $G$  its uniform upper density by
$$
\udens(\vL)=\limsup_n  \sup_{x\in G} \frac{\card (\vL\cap (A_n +x))}{|A_n|} \ .
$$
The uniform upper density can be seen to be independent of the van Hove sequence \cite{LSS,PRS22}.  However, we do not need this and refrain from further discussion.

\begin{theorem}[Characterization of regular model sets]\label{thm:char-mod-set} Let $\vL$ be a Meyer set in $G$.  Then the following assertions are equivalent.
\begin{itemize}
\item[(i)]  $\vL$ is a regular model set.

\item[(ii)] There exists a Meyer set $\vG\supset\vL$ with the following property:  For every $\eps>0$ there are strongly almost periodic measures $\mu_\eps, \nu_\eps$ with
$$
0\le \mu_\eps \le \delta_{\vL} \le \nu_\eps \leq \delta_{\vG} \mbox{ and }
M(\nu_\eps - \mu_\eps)\leq  \udens(\{\mu_\eps \neq \nu_\eps\}) < \eps.$$

\item[(iii)] $\vL$ is an almost periodic pattern.

\end{itemize}
\end{theorem}

\begin{remark}(a) For Euclidean space $G$, the implication (i) $\Rightarrow$ (iii) has already been established by Meyer in \cite[Cor.~4.3]{Mey2}.

(b)  Note that we have $$\{\mu_\eps \neq \nu_\eps\}= \{ \delta_{\vL} \neq \mu_{\eps}\} \cup \{ \delta_{\vL} \neq \nu_{\eps}\} \ .$$
Therefore, (ii) asserts that any regular model set $\vL$ can be approximated from above and below by discrete strongly almost periodic measures that only differ from $\delta_\vL$ on sets of arbitrary small density.  This is a precise  version of a viewpoint sometimes  expressed by  physicists (see  the introduction).
\end{remark}

\begin{proof}[Proof of Theorem~\ref{thm:char-mod-set}]
\noindent (i) $\Longrightarrow$ (ii): As $\vL$ is a regular model set, we can find a cut-and-project scheme
$S=(H,\cL)$ and a regular window  $W\subset H$ such that $\vL = \oplam (W)$. The following argument only uses the inclusions
$\oplam (W^\circ)\subset \vL \subset \oplam (\overline{W})$.

Fix a compact set $C\subset H$ such that $\overline{W} \subset C^\circ$ and define the Meyer set $\vG=\oplam(C)$. Let $\eps>0$ be arbitrary. As the Haar measure $\theta_H$ on $H$ is a Radon measure, the regularity properties of $\theta_H$ together with $\theta_H(\partial W)=0$ imply that there exist a compact set $K\subset H$ and an open set $O'\subset H$ such that $K \subset W^\circ\subset \overline{W}\subset O'$ and $\theta_H(O'\setminus K)\le \eps$.
Then the open set $O= O' \cap C^\circ$ also satisfies $K \subset W^\circ\subset \overline{W}\subset O$ and $\theta_H(O\setminus K)\le \eps$.

By \cite[Thm.~2.7]{Rud} and by $W^\circ\ne\varnothing$ we find nonempty open sets $U,V\subset H$ with compact closure such that $K \subset U \subset \overline{U}\subset W^\circ\subset \overline{W} \subset V \subset \overline{V} \subset O$.
By Urysohn's lemma, we find $f,g\in \Cc(H)$ such that
$1_{\overline{U}} \le f \le 1_{W^\circ} \le 1_{\overline{W}} \le g \le 1_{V}$.
Then $\mu=\vOmega(f)$ and $\nu=\vOmega(g)$ belong to $\SAP (G)$ by Lemma \ref{omega-is-sap}. Moreover both $\mu$ and $\nu$ have Meyer set support and satisfy
\begin{displaymath}
0\le\mu\le \delta_{\oplam(W^\circ)} \le \delta_\vL\le \delta_{\oplam(\overline{W}) }\le \nu\le \delta_{\vG} \ .
\end{displaymath}

Note that for any $h\in \Cc (H)$ with $1_{\overline{V}\setminus U} \leq h$   we have  the estimates
\begin{displaymath}
 \card(\{\mu  \neq \nu \} \cap A)\le \card(\oplam(\overline{V}\setminus U) \cap A)\le \sum_{x\in L\cap A } h(x^\star) = \vOmega(h)(A)
\end{displaymath}
on any compact $A\subset G$.
From  Lemma~\ref{computing-the-mean} and Proposition~\ref{prop:meanav} we therefore find
$$\udens(\{\mu\neq \nu\})  \leq M(\vOmega(h)) = D_S \cdot \theta_H(h) \ .
$$
As $\overline{V}\setminus U$ is compact in $H$, we can choose  $h\in \Cc (H)$ with $1_{\overline{V}\setminus U} \leq h$   such that $\theta_H(h)$ is arbitrarily close to $\theta_H(\overline{V}\setminus U)$. Hence, we find
$$\udens(\{\mu\neq \nu\}) \leq   D_S \cdot \theta_H(\overline{V}\setminus U)
    \le D_S \cdot \theta_H(O\setminus K)\le D_S \cdot \eps \ .$$
We finally note  $0 \leq \nu-\mu \leq \delta_{\{\mu\neq \nu\}}$, which implies
$$M(\nu-\mu)  \leq \udens(\{\mu\neq \nu\})\ .$$

\smallskip

\noindent (ii) $\Longrightarrow$ (iii): This is obvious.

\smallskip

\noindent (iii) $\Longrightarrow$ (i):
By Proposition~\ref{prop:excBc}, we may choose a universal cut-and-project scheme $S=(H,\cL)$ to which the Meyer set $\vL$ is associated. Then $\delta_\vL\in \cM_S(G)$, and we also have  $\delta_\vL\in \GAP(G)$ as $\vL$ is an almost periodic pattern by assumption.
Now we invoke Theorem~\ref{thm:char-weighted-model} to obtain a Riemann integrable $h: H \to \CC $ with $\delta_\vL = \vOmega(h)$. Note that we have  $h \in \{0,1 \}$ on $L^\star$, which is a dense set in $H$.
Thus, by Lemma~\ref{lem:2_WRI}, there exists  $W\subset H$ having topological boundary of Haar measure zero such that
\[
\delta_{\vL} = \vOmega(h)  = \vOmega(1_W) = \delta_{\oplam(W)} \ ,
\]
and $\vL = \oplam(W)$ follows.

Moreover, the interior of $W$ is not empty as $\vL$ is relatively dense. Indeed, note that for any  $g\in \Cc(H)$ such that $1_W\le g$ we have $\udens(\vL)\le M(\vOmega(g))=D_S\cdot \theta_H(g)$. Now $W^\circ=\varnothing$ would imply $\theta_H(W)=0$ as $\theta_H(\partial W)=0$. Hence we can choose $g\in\Cc(H)$ with $1_W\le g$ such that $\theta_H(g)$ is arbitrarily close to $0$. Thus $\udens(\vL)=0$, which is contradictory as $\vL$ has positive upper density by relative denseness.

Altogether, we have thus shown that $\vL$ is a regular model set.
\end{proof}

\section{Examples and counter-examples}\label{sec-ex}
This section gathers various remarks and examples to illuminate  facets of our results.

\subsection{Meyer sets in a lattice}\label{sec:discrete}

Let $L$ be a lattice in $G$. In order to describe arbitrary Meyer sets in $L$ in a common cut-and-project scheme, we may use the Bohr compactification $L_\mathsf{b}$ of $L$ with injection map $i:L\to L_\mathsf{b}$.

\begin{proposition}[A cut-and-project scheme for lattice subsets]\label{prop:discrete}
Let $L$ be a lattice in $G$ and define $\cL=\{(x, i(x)): x\in L\}\subset G\times L_\mathsf{b}$. Then $S=(L_\mathsf{b}, \cL)$ is a cut-and-project scheme over $G$ with injective star map, to which $L$ is associated. Moreover the spaces $\cM_{S}(G)$, $\cM^\infty(L)$, and $\Cu(L)$ are canonically isomorphic.
\end{proposition}
\begin{proof} We note that the Meyer set $L\subset G$ is associated to the cut-and-project scheme  $S'=(\{0\}, L\times \{0\})$ over $G$. The construction in the proof of Proposition~ \ref{prop:excBc}, with $G_{\mathsf{b}}$ replaced by $L_\mathsf{b}$, then produces the compatible cut-and-project scheme $S$. The injection maps $\Cu(L)\to \cM^\infty(L)\to \cM_S(G)$ given by $f \mapsto f\cdot \theta_L\mapsto f\cdot \delta_L$ are isomorphisms as $L$ is uniformly discrete in $G$. \end{proof}

\begin{remark}
(a)
For discrete $G$, we may take $L=G$. In that case, the above cut-and-project scheme is universal. In fact, $\GAP(G)$ can then be identified with a class of functions that are called Riemann almost periodic by Hartman \cite{H64}.

(b)
The above cut-and-project scheme may not be universal, but allows to lift all functions that are Bohr almost periodic on $L$. This has been used in \cite[Sec.~6]{LLRSS} to analyse modulated lattices, which give rise to g-a-p measures beyond having Meyer set support.
\end{remark}

\subsection{On regular windows}\label{sec:regular}
A model  set may well be associated to   more than one cut-and-project scheme. Then, the question arises whether regularity of the window is independent of the chosen cut-and-project scheme. Our considerations above give the following.

\begin{proposition}\label{prop:unireg} Let $\vL$ be a Meyer set in $G$, and let $S=(H,\cL)$ be any $\Bohr$ cut-and-project scheme over $G$ to which $\vL$ is associated.
Then, $\vL$ is a regular model set if and only if there exists a regular window $W \subset H$ such that $\vL = \oplam(W)$.
\end{proposition}

\begin{proof}
\noindent \emph{``if part''}: This holds by definition.

\noindent \emph{``only if part''}: As $\vL$ is a regular model set, its Dirac comb is g-a-p by  Theorem~\ref{thm:char-mod-set}. The proof (iii) $\Rightarrow$ (i) of Theorem~\ref{thm:char-mod-set} then gives  $\vL=\oplam(W)$ for a regular window $W\subset H$ in the universal cut-and-project scheme $S$.
\end{proof}

The proposition asserts that regularity necessary holds with respect to $\Bohr$ cut-and-project schemes. However the following examples in $G=\ZZ$ show that regularity may be lost by passing to arbitrary cut-and-project schemes over $G$:

\smallskip

Consider for any prime $p$ the group of $p$-adic integers $H_p=\ZZ_p$ together with $\cL_{p}= \{ (n, i_p(n)) : n \in \ZZ \}$, where $i_p : \ZZ \to \ZZ_p$ denotes the canonical embedding. Define $H_0= \{ 0 \}$ and $\cL_0= \{ (n,0) : n \in \ZZ \}$.  Then $S_0=(H_0, \cL_0)$ and $S_p=(H_p, \cL_p)$ are cut-and-project schemes over $\ZZ$. Moreover, a Meyer set in $\ZZ$ is associated to $S_0$ or $S_p$, respectively, if and only if it is a relatively dense subset of $\ZZ$. Therefore, all these cut-and-project schemes are compatible. Now the following holds true.
\begin{itemize}
  \item{} The only cut-and-project sets from $S_0$ are $\vL=\varnothing$ and the regular model set $\vL=\ZZ$. Every subset of $\ZZ$ is a cut-and-project set from $S_p$, as $S_p$ has an injective star map $i_p$.
  \item{} Consider $\vL_{p,r} = p \ZZ +r$ for some prime $p$ and some $r \in \{0,1, \ldots, p-1 \}$. Then $\vL_{p,r}=\oplam_{S_p}(W_{p,r})$ is a regular model set, where $W_{p,r} \subset H_p$ is the closed ball of radius 1 centered at $i_p(r)$.
  \item{} Let now $q$ be a prime different from $p$. Then no window $W\subset H_q$ such that $\vL_{p,r}=\oplam_{S_q}(W)$ can be regular. Indeed, otherwise $\vL_{p,r}$ would have density 1 as $i_q(\vL_{p,r})$ is  dense in $H_q$, which is contradictory.
\end{itemize}

\subsection{Regular inter model sets}

When studying Pisot substitutions or model set dynamical systems, so-called
regular inter model sets play an important r\^ole \cite{BLM,Lee,RS24}. These are Meyer sets $\vL$ satisfying $\oplam(W^\circ)\subset \vL\subset \oplam(\overline{W})$ in some cut-and-project scheme with regular window $W$.

Clearly any regular model set is a regular inter model set.
Moreover, observe that the proof (i) $\Rightarrow$ (ii) in Theorem~\ref{thm:char-mod-set} carries over verbatim from regular model sets  to  regular inter model sets. Thus generalized almost periodicity also characterizes regular inter model sets and we obtain the following statement.

\begin{proposition}
A subset of $G$ is a regular inter model set if and only if it is a regular model set. \qed
\end{proposition}

For Euclidean space $G$, this has recently been established in \cite{RS24} by a geometric approach that does not seem to be extendable to arbitrary locally compact abelian groups.

\subsection{G-a-p measures and Weyl almost periodicity}

As discussed in the introduction, Weyl almost periodic measures constitute an important subclass of models in aperiodic order. It is shown in \cite{LSS} that any g-a-p measure is Weyl almost periodic. Here, we show that the converse does not hold. So, g-a-p is a strictly stronger notion that Weyl almost periodicity.  In fact there is a crucial difference between these two concepts: Weyl almost periodicity is stable under perturbations of zero density while g-a-p is not. One might even think of g-a-p measures as (almost) being locally determined.

\smallskip

Recall that a locally integrable function $f \in L^1_{loc}(G)$ is Weyl almost periodic if for every $\eps >0$ there exists a trigonometric polynomial $P_\eps$ such that
\[
\| f- P_\eps \|_{w}= \limsup_{n} \sup_{x \in G} \frac{1}{|A_n|} \int_{A_n} \left|f(x-t)-P_\eps(x-t) \right| \dd t <\eps \ ,
\]
where $(A_n)$ is any van Hove sequence. In fact the above limit is independent of the choice of the van Hove sequence \cite[Prop.~4.11]{LSS}.
A measure $\mu\in \cM^\infty(G)$ is \textit{Weyl almost periodic} if the function $\varphi\ast \mu$ is Weyl almost periodic for all $\varphi\in \Cc(G)$.

Denote the total variation of the translation bounded measure $\mu$ by $|\mu|$ and define
its \textit{uniform upper density} by
\[
\|\mu \|_{u}  = \limsup_{n} \sup_{x\in G} \frac{|\mu|(x+A_n)}{|A_n|} \ .
\]
This is independent of the choice of the van Hove sequence \cite[Prop.~5.14]{PRS22}. Also note $\|\delta_\vL\|_{u}=\udens(\vL)$ for any uniformly discrete point set $\vL$.
The following is immediate.

\begin{proposition}[Stability of Weyl almost periodicity under zero uniform density perturbations]\label{prop:wap}
If $\mu$ is Weyl almost periodic and $\nu\in \cM^\infty(G)$ is such that $\| \mu -\nu \|_{u}=0$ then $\nu$ is also Weyl almost periodic. \qed
\end{proposition}

The stability property given in the previous proposition does not hold for g-a-p measures. We rather have the following result.

\begin{theorem}\label{example} There exists a relatively dense point set $\vL \subset \ZZ$ such that $\| \delta_{\vL}- \delta_{\ZZ} \|_{u}=0$, and $\vL$ is not a regular model set in $\ZZ$.
In particular, $\delta_{\vL} $ is Weyl almost periodic but not generalized-almost-periodic.
\end{theorem}

We will  prove the theorem  building on an example due to Meyer \cite[Sect.~2]{Mey2}.  This requires some preparation. For any transcendental number $\theta >3$ consider
\begin{displaymath}
\vL_\theta = \left\{ \sum_{j=0}^n  c_j \theta^j  : n \in \NN, c_0,\ldots, c_j \in \{0,1\} \right\} \ .
\end{displaymath}
The set $\vL_\theta$ gives rise to a set of non-negative integers $M_\theta$ via
\begin{displaymath}
M_\theta = \left\{\lfloor x \rfloor, \lfloor x \rfloor+1 : x \in \vL_\theta \right\} \ ,
\end{displaymath}
where $\lfloor x \rfloor$ denotes the largest integer smaller or equal to $x$. The following version of \cite[Lem.~2.30]{Mey2} is adapted to our needs.
\begin{lemma}\label{prop:mey} Assume that $\theta>3$ is transcendental. Then $M_\theta$ has uniform density $0$, and $M_\theta$ is dense in the Bohr compactification of $\ZZ$.
\end{lemma}

\begin{proof}
Denseness of $M_\theta$ in the Bohr compactification of $\ZZ$ is proved in \cite[Lem.~2.36]{Mey2}. We show that $M_\theta$ has indeed uniform density $0$ by arguing that $\vL_\theta$ has uniform density $0$.  By transcendence, each nonzero $x \in \ZZ[\theta]$ can be written as $x =P(\theta)$ for a unique non-zero polynomial $P$ having integer coefficients. Denote its degree by $d(x)$ and its leading coefficient by $c(x)$. Consider now
\[
\vG_\theta=\vL_\theta-\vL_\theta= \left\{ \sum_{j=0}^n c_j \theta^j : n \in \NN, c_j \in \{-1,0,1\} \right\}  \ .
\]
Observe that we have for all $n\in \NN$ the estimate
\[
1+\theta+\ldots +\theta^{n-1} < \theta^{n} \ ,
\]
where we used $\theta>1$. Thus $x\in \vG_\theta$ positive implies $c(x)=1$, due to $0< x< \theta^{d(x)}+c(x)\theta^{d(x)}$.
Furthermore, we have for all $n\in\NN$ the estimate
\[
1+\theta+\ldots +\theta^{n-1} < \theta^{n}-(1+\theta+\ldots +\theta^{n-1} )\ ,
\]
where we used $\theta>3$.
Indeed, the above inequality is equivalent to $\theta^{n}(\theta-3)+2>0$. Thus if positive numbers $x,y\in \vG_\theta$ satisfy $x\le y$, then we have $d(x)\le d(y)$. This holds as $d(x)>d(y)$ implies
\begin{displaymath}
x-y\ge \theta^{d(x)}- \sum_{j=0}^{d(x)-1} \theta^j -
 y>  \sum_{j=0}^{d(x)-1} \theta^j - \sum_{j=0}^{d(y)} \theta^j \ge0 \ .
\end{displaymath}
We  now estimate the number of elements in  $F=\vL_\theta \cap [t,t+n]$ for fixed $t \geq 0$ and arbitrary $n\in\NN$. Denote by $\alpha$ the smallest and by $\beta$ the largest element in $F$. Then $\beta - \alpha \leq n$ and
\[
F-\alpha \subseteq \vG_\theta \cap [0, \beta-\alpha] = \{ 0\} \cup \left( \vG_\theta \cap (0, \beta-\alpha] \right) \ .
\]
Now, for all $x \in \vG_\theta \cap (0, \beta-\alpha]$ we have $d(x) \leq d(\beta-\alpha)$. Writing $d= d(\beta-\alpha)$, we get
\[
\vG_\theta \cap (0, \beta-\alpha] \subseteq \{c_0+c_1\theta+ \ldots + c_{d-1}\theta^{d-1}+c_d \theta^d : c_0, \ldots c_d \in \{-1,0,1\} \} \ .
\]
This shows that $\card(F) = \card(F-\alpha) \leq 3^{d+1}+1$. Now a standard estimate yields
\[
n \ge \beta-\alpha \ge \theta^d -(1+\theta+\ldots+\theta^{d-1}) > 1+\theta+\ldots+\theta^{d-1} > \theta^{d-1} \ ,
\]
and hence $\theta^{d-1} < n$. We thus have $d < 1+\log_\theta(n)$.  This gives the estimate
\[
\frac{\card(\vL_\theta \cap [t,t+n])}{n}< \frac{3^{2+\log_\theta(n)}+1}{n}=\frac{9 n ^{\log_\theta(3)}+1}{n} \ .
\]
Note that when $t <0$, due to $\vL_\theta \subset [0, \infty)$ we have $\vL_\theta \cap [t,t+n] \subset \vL_\theta \cap [0,n]$. Thus the above estimate holds in fact for all $t\in \RR$ and all $n\in \NN$. As $\theta>3$, this shows that $\vL_\theta$ has indeed uniform density $0$.
\end{proof}

We can now show that the point set $\ZZ\setminus M_\theta$ has the properties stated in Theorem~\ref{example}.

\begin{proof}[Proof of Theorem~\ref{example}]
Assume that $\theta>3$ is transcendental and define $\vL= \ZZ \backslash M_\theta$. Then $\delta_{\ZZ} - \delta_{\vL}=\delta_{M_\theta}$ and hence $\|\delta_{\vL}- \delta_{\ZZ} \|_{u}=0$.
Thus Weyl almost periodicity of $\vL$ holds by Proposition~\ref{prop:wap}. Moreover relative denseness of $\vL$ holds as $M_\theta$ has uniform density $0$.

Now assume by contradiction that $\vL\subset \ZZ$ is a regular model set. Then, by Proposition~\ref{prop:discrete}, $\vL$ is a cut-and-project set in the $\Bohr$ cut-and-project scheme $S=(\ZZ_{\mathsf{b}}, \cL)$ over $\ZZ$. Therefore, by Proposition~\ref{prop:unireg}, there exists a regular window $W \subset \ZZ_{\mathsf{b}}$ such that $\vL = \oplam_{S}(W)$.
Since $i(\ZZ \backslash \vL)$ is dense in $\ZZ_{\mathsf{b}}$, the regular window $W$ must have empty interior, which is contradictory.
\end{proof}

\begin{remark} (a)
The point set $M_\theta$ is Weyl almost periodic, as it has uniform density $0$. Moreover $M_\theta=\oplam(W)$ in the universal cut-and-project scheme from Proposition~\ref{prop:discrete} for some windows $W \subset \ZZ_{\mathsf{b}}$. Any such window satisfies $W^\circ=\varnothing$ due to density $0$ and $\overline{W}=\ZZ_{\mathsf{b}}$ by Lemma~\ref{prop:mey}.

(b) The Weyl almost periodic point set $\ZZ\setminus M_\theta$ is not a regular model set, but agrees with the regular model set $\ZZ$ up to a set of uniform density $0$. In particular, it has the same diffraction properties as $\ZZ$.
\end{remark}

\section{More about $\Bohr$ cut-and-project schemes}\label{sec-Bohr}

In this section we have a closer  look at $\Bohr$ cut-and-project schemes. We show uniqueness of compatible $\Bohr$ cut-and-project schemes up to a suitable notion of isomorphism, discuss universality in terms of a  factorization property on compatible cut-and-project schemes, and provide a number of characterizations.

\subsection{Uniqueness and factorization}\label{sec:uu}

We will use the following notion of isomorphy of cut-and-project schemes.
\begin{definition}[Isomorphy of cut-and-project schemes]
Consider two cut-and-project schemes  $S=(H,\cL)$ and $S'=(H',\cL')$ over $G$. Then $S$ and $S'$ are called \textit{isomorphic} if there exists a group isomorphism $\Phi:H'\to H$ that is a homeomorphism such that $\cL=\{(x,\Phi(y')): (x,y') \in \cL'\}$.
\end{definition}

Universal cut-and-project schemes are uniquely determined in the following sense.

\begin{proposition}[Uniqueness of a $\Bohr$ cut-and-project scheme up to isomorphism]\label{prop:uni}
Any two  compatible $\Bohr$ cut-and-project schemes are isomorphic.
\end{proposition}

\begin{proof}
Let $S=(H,\cL)$ and $S'=(H',\cL')$ be two compatible $\Bohr$ cut-and-project schemes. Denote the star maps of $S, S'$ by $\star,\star'$, respectively, and note $L=\pi^G(\cL)=\pi^G(\cL')$ by compatibility. As the star maps are injective by universality, there exists a group isomorphism $\phi:L^{\star'}\to L^{\star}\subset H$.
Equip $L^{\star'}\subset H'$ with the topology inherited from $H'$.  We show continuity of $\phi: L^{\star'}\to H$ by arguing that $\phi^{-1}(U)$ is open for all nonempty open $U\subset H$. Consider any such $U$. Take any $y\in U$ and any $V$ with compact closure such that $y\in V\subset U$. Invoke Urysohn's lemma to find $h\in\Cc(H)$ satisfying $0\le h\le 1$ such that $h(y)=1$ and $h=0$ outside $V$. Consider the set $U_y=\{h>0\}$ which is open in $H$.  For the strongly almost periodic measure $\mu=\vOmega_{S}(h)$ we use Corollary~\ref{cor-rep} to find $h'\in \Cc(H')$ such that $\mu=\vOmega_{S'}(h')$. We then have
\begin{align*}
\phi^{-1}(U_y)&=\phi^{-1}(\{h>0\})
= \phi^{-1}(\{x^{\star}: \mu(\{x\})>0\}) \\
&= \{x^{\star'}: \mu(\{x\})>0\} = \{h'>0\} \cap L^{\star'} \ ,
\end{align*}
which shows that $\phi^{-1}(U_y)$ is open in $L^{\star'}$. As $y\in U$ was arbitrary and satisfies $y\in U_y\subset U$, we conclude that $\phi:L^{\star'}\to H$ is continuous.
Since $\phi$ is a group homomorphism, it is uniformly continuous, and hence, by denseness of $L^{\star'}$ in $H'$ it extends uniquely to a continuous group homomorphism $\Phi : H' \to H$.
We thus have  $\Phi(x^{\star'})= \phi(x^{\star'})=x^{\star}$ for all $x \in L$.
Repeating the argument with $H$ and $H'$ interchanged, we infer the existence of a continuous group homomorphism $\Psi : H \to H'$ satisfying
$\Psi(x^{\star})=x^{\star'}$. This implies  $\Psi \circ \Phi=\mbox{Id}$ on $L^{\star'}$ and $\Phi\circ\Psi =\mbox{Id}$ on $L^{\star}$. By denseness and continuity, we get that $\Phi$ is a homeomorphism with inverse $\Psi$, and that $\Phi$ is a group isomorphism. We finally note
\begin{displaymath}
\cL=\{(x,x^{\star}): x\in L\}=\{(x,\phi(x^{\star'})): x\in L\}=\{(x,\Phi(x^{\star'})): x\in L\}=\cL' \ .
\end{displaymath}
This shows that $S$ and $S'$ are isomorphic.
\end{proof}

For our discussion of universality, we will use the following notion of factor of a cut-and-project scheme.

\begin{definition}[Factor cut-and-project scheme]
Consider two cut-and-project schemes  $S=(H,\cL)$ and $S'=(H',\cL')$ over $G$. Then $S$ is called a \textit{factor} of $S'$ if there exists a continuous  group homomorphism $\Psi:H'\to H$ such that $\cL=\{(x,\Psi(y')): (x,y') \in \cL'\}$.
\end{definition}
We note that such a continuous group homomorphism must necessarily be unique  (by denseness of  the range of the  projection of  $\cL'$ to $H'$). From the preceding proposition we obtain the following consequence.

\begin{proposition}[Factorization over universal cut-and-project schemes]
Let $S'$ be a universal cut-and-project scheme over $G$. Then  $S'$ has the factorization property that  any compatible cut-and-project scheme $S$ over $G$ is a factor of $S'$.
\end{proposition}

\begin{proof}
Consider the cut-and-project scheme  $S_\mathsf{u}=(H_\mathsf{u}, \cL_\mathsf{u})$ over $G$ that was constructed from $S$ in order to prove Proposition ~\ref{prop:excBc}. Recall $H_\mathsf{u}\subset H\times G_\mathsf{b}$ and denote by $\pi_1:H_\mathsf{u}\to H$ the canonical projection on the first factor. Observe that $S_\mathsf{u}$ and $S'$ are compatible, since $S$ and $S_\mathsf{u}$ are compatible by construction and since $S$ and $S'$ are compatible by assumption. As $S_\mathsf{u}$ is universal, there exists a continuous group homomorphism $\Phi:H'\to H_\mathsf{u}$ such that $\cL_\mathsf{u}=\{(x,\Phi(y')): (x,y') \in \cL'\}$ by Proposition~\ref{prop:uni}. Define the continuous group homomorphism $\Psi: H'\to H$ by $\Psi=\pi_1 \circ \Phi$. We then have
\begin{displaymath}
\begin{split}
\cL&=\{(x,y): (x,y) \in \cL\}=\{(x,\pi_1(y_\mathsf{u})): (x,y_\mathsf{u}) \in \cL_\mathsf{u}\}\\
&=\{(x,\pi_1(\Phi(y'))): (x,y')\in \cL'\}= \{(x, \Psi(y')):(x,y')\in \cL'\} \ .
\end{split}
\end{displaymath}
We have shown that $S$ is a factor of $S'$.
\end{proof}

\begin{remark} It is not hard to see that compatible  universal cut-and-project schemes are characterized by the factorization property. Indeed, let $S=(H,\cL)$ and $S'=(H',\cL')$ over $G$ be compatible and  both have the factorization property. Let  $\Psi : H'\longrightarrow H$ and $\Phi : H\longrightarrow H'$ be the arising factor maps. Then, for any $(x,y)\in\mathcal{L}$ we find that $(x,\Psi \circ \Phi (y))$ belongs to $\mathcal{L}$. This gives $y = x^\star = \Psi \circ \Phi (y)$. Hence, $\Psi \circ \Phi$ is the identity on $H$  (first on a dense subset and then, by continuity, on the whole of $H$).  Similarly, we find $\Phi \circ \Psi$ is the identity on $H'$. This shows that $S$ and $S'$ are isomorphic.
\end{remark}

\subsection{Characterizations}

Here we provide a number of characterizations of $\Bohr$ cut-and-project schemes. For the following statement, recall that the dual group $\widehat{A}$ of a locally compact abelian group $A$ consist of all continuous group homomorphisms from $A$  to the group $U(1) =\{z\in\CC: |z|=1\}$ (equipped with multiplication).

\begin{theorem}[Characterisations of the $\Bohr$ cut-and-project scheme]\label{thm:char-bu}
Let $S=(H, \cL)$ be a cut-and-project scheme over $G$ with associated map $\vOmega$. Let $\tau$ be the topology on $L=\pi^G(\cL)$ induced by $L \stackrel{\star}{\longrightarrow} H$. Then, the following assertions  are equivalent for $S$.
\begin{itemize}
\item[(i)] $S$ is a $\Bohr$ cut-and-project scheme over $G$.

 \item[(ii)] There exists a map $\left( \right)^\sharp: SAP(G)\to \Cb(H)$
such that for all $f \in SAP(G)$ and $h \in \Cc(H)$ we have $f\cdot \vOmega(h)= \vOmega(f^\sharp \cdot h)$.

\item[(iii)] There exists a (non-necessarily continuous) group homomorphism $\left( \right)^\sharp: \widehat{G} \to \widehat{H}$ such that for all $\chi \in \widehat{G}$ we have $\chi(x)= \chi^\sharp(x^\star)$ for all $x \in L$.

\item[(iv)] There exists a map $\left( \right)^\sharp: \widehat{G} \to \widehat{H}$ such that for all $\chi \in \widehat{G}$ we have
  $\chi(x)= \chi^\sharp(x^\star)$ for all $x \in L$.

\item[(v)] The map $$\Cc(H)\longrightarrow
\SAP(G)\cap \mathcal{M}_S (G) \ , \quad  h\mapsto \vOmega (h) \ ,$$ is a bijection.

\item[(vi)] There exists a continuous group homomorphism $F : H \to G_{\mathsf{b}}$ such that the following diagram commutes:
\[
  \begin{tikzcd}
    L \arrow[r,"\star"] \arrow[d, hook] & H \arrow[d, dashed, "F"] \\
     G\arrow[r,hook, "i_{\mathsf{b}}"] &  G_{\mathsf{b}}
  \end{tikzcd}
\]

\item[(vii)] The map $\Psi: L \to G_{\mathsf{b}}$ given by $\Psi(x)=i_{\mathsf{b}}(x)$ is uniformly continuous with respect to the topology $\tau$.
\end{itemize}
Moreover, if one of the equivalent conditions (i),$\ldots$,(vii) holds,   the following are also valid:
\begin{itemize}
	\item[(a)] The  map $f \mapsto f^\sharp$ is linear and positive, and
  $\| f^\sharp \|_\infty = \| f|_{L}\|_\infty \leq \|f \|_\infty$ holds, where $|_L$ denotes restriction to $L$.
  \item[(b)] For all $f \in SAP(G)$ the function
	$f^\sharp $ belongs to $ SAP(H)$.
\end{itemize}
\end{theorem}

\begin{proof}

Property (a) is immediate from the definition.

\smallskip

\noindent (i) $\Longrightarrow$ (ii):  This is shown after Definition~\ref{def:bc}.

\smallskip

\noindent (i) $\Longrightarrow$ (v):  This follows from Corollary~\ref{cor-rep}, as any cut-and-project scheme is compatible with itself.

\noindent (v) $\Longrightarrow$ (i): Let $\{ h_i: i\in I \}$ be a partition of unity of $H$, consisting of nonnegative continuous functions of compact support. Let $f \in SAP(G)$ be arbitrary. For each $i\in I$ we have $\vOmega(h_i) \in  \SAP(G)\cap  \cM_{S}(G)$. Hence by assumption there exists some  $f_i \in \Cc(H)$ such that
\[
f\cdot \vOmega(h_i) = \vOmega(f_i) \,.
\]
Define $f^\sharp= \sum_{i\in I} f_i$.
We show that this mapping satisfies the desired conditions. First, the density of the star map gives
\[
\supp(f_i)= \overline{\supp(\vOmega(f_i))^\star}= \overline{\supp(f\cdot \vOmega(h_i))^\star} \subset \overline{\supp(\vOmega(h_i))^\star} =\supp(h_i) \,.
\]
Since $\{h_i: i \in I\}$ is a partition of unity, each point $y \in H$ has an open neighbourhood in which all but finitely many $h_i$ are zero. This immediately implies that
$f^\sharp$ is well defined and continuous at each $y \in H$.
Finally, for all $x \in L$ we have
\[
f^\sharp(x^\star)= \sum_{i} f_i(x^\star)= \sum_{i\in I} f(x) \cdot \vOmega(h_i)(\{x\})= f(x) \sum_{i\in I} h_i(x^\star)=f(x) \,.
\]

\noindent (ii) $\Longrightarrow$ (i): This can be shown using a partition of unity argument as in ``(v) $\Rightarrow$ (i)''.

\noindent (i) $\Longrightarrow$ (iii):  Let $\chi \in \widehat{G}$ be arbitrary. As $\chi\in SAP(G)$, we infer that there exists some $\chi^\sharp \in \Cb(H)$ satisfying
$\chi(x)= \left(\chi \right)^\sharp(x^\star)$ for all $x\in L$. Continuity of $\chi^\sharp$ and denseness of $L^\star$ in $H$ give $\chi^\sharp(H) \subset U(1)$.
Next, for all $x,y \in L$ we have
\[
\chi^\sharp(x^\star+y^\star)=\chi(x+y)=\chi(x)\cdot\chi(y)=\chi^\sharp(x^\star) \cdot \chi^\sharp(y^\star)\,.
\]
Continuity of $\chi^\sharp$ and denseness of $L^\star$ imply then that $\chi^\sharp \in \widehat{H}$. Next, for all $\chi, \eta \in \widehat{G}$ and all $x \in L$ we have
\[
\left( \chi \cdot \eta \right)^\sharp(x^\star)= (\chi \cdot \eta)(x)=\chi(x) \cdot \eta(x)= \chi^\sharp(x^\star)\cdot \eta^\sharp(x^\star) \ .
\]
Again, the denseness of $L^\star$ and continuity of characters imply that $\left( \right)^\sharp: \widehat{G} \to \widehat{H}$ is a group homomorphism.

\smallskip

\noindent (iii) $\Longrightarrow$ (iv): This is obvious.

\smallskip

\noindent (iv) $\Longrightarrow$ (i): Let $f\in SAP(G)$. Then for every $n\in\NN$ there exists a trigonometric polynomial $P_n=\sum_j a_{j,n} \chi_{j,n}$ such that $\| f -P_n \|_\infty < 1/n$. Define $Q_n= \sum_j a_{j,n} \chi^\sharp_{j,n}$ and note $Q_n \in \Cu(H)$ as $\chi^\sharp_{j,n} \in \Cu(H)$.
Then $Q_n(x^\star)=P_n(x)$ for all $x \in L$. Again, using the denseness of $L^\star$, we have
\[
\|Q_m-Q_n \|_\infty \leq \|P_m-P_n \|_\infty \ .
\]
This shows that the sequence $(Q_n)$ is Cauchy in $(\Cu(H), \| \cdot \|_\infty)$ and hence convergent to some $f^\sharp \in  \Cu(H)$. Moreover $f^\sharp\in SAP(H)$ as it is the uniform limit of the trigonometric polynomials $(Q_n)$. It is immediate that $f^\sharp(x^\star)=f(x)$ for all $x \in L$. This proves (i) and (b).

\smallskip

\noindent (iii) $\Longrightarrow$ (vi):
Equip $\widehat{G}$ with the discrete topology and call the resulting group $\widehat{G}_d$.
Then, the map $( )^\sharp : \widehat{G}_{d} \to \widehat{H}$ is a continuous group homomorphism by assumption. Therefore, by Pontryagin duality, its dual map
$F= \reallywidehat{ ( )^\sharp}: H \to G_{\mathsf{b}} $ is a continuous group homomorphism, where $G_{\mathsf{b}}=\reallywidehat{\widehat{G}_d}$ is the Bohr compactification of $G$. Note that we have for every $y\in H$ and for every $\chi\in\widehat{G_{\mathsf{b}}}$ that $F(y)(\chi)=\chi^\sharp(y)$.
In particular, for all $x \in L$ we have
\[
F(x^\star)(\chi)= \chi^\sharp(x^\star)=\chi(x)=i_{\mathsf{b}}(x)(\chi) \,,
\]
with the last equality holding under the identification $\widehat{G_{\mathsf{b}}}=\widehat{G}$. Since this holds for all $\chi \in \widehat{G}=\widehat{G_{\mathsf{b}}}$ we get $F(x^\star)=i_{\mathsf{b}}(x)$. Thus the diagram is commutative.

\smallskip

\noindent (vi) $\Longrightarrow$ (iii): For all $\chi \in \widehat{G}_d= \widehat{G_{\mathsf{b}}}$ and $x \in L$ we have by the definition of the dual operator:
\[
\reallywidehat{F}(\chi)(x^\star)=\chi(F(x^\star))= \chi(i_{\mathsf{b}}(x))= \chi(x) \,.
\]
This means that the mapping $(\chi)^\sharp:= \reallywidehat{F}(\chi)$ satisfies (iii).

\smallskip

\noindent (vi) $\Longrightarrow$ (vii): Since $F$ is continuous and a group homomorphism, it is uniformly continuous. Since $\star: (L, \tau) \to H$ is also uniformly continuous, so is $\Psi= F \circ \star$.

\smallskip

\noindent (vii) $\Longrightarrow$ (vi): Note here that $(H, \star)$ is a completion of $(L, \tau)$. Since $\Psi$ is uniformly continuous, it extends uniquely to a continuous mapping $F: H \to G_{\mathsf{b}}$ which trivially satisfies the given conditions.
\end{proof}

\subsection*{Acknowledgments}
We thank Michael Baake for suggesting the upper density estimate in (ii) of Theorem~\ref{thm:char-mod-set}.
Part of this work was done when NS visited the Friedrich-Schiller Universit\"at Jena and the FAU Erlangen-N\"urnberg, and he would like to thank the Mathematics Departments for hospitality.
NS was supported by the Natural Sciences and Engineering Council of Canada via grants 2020-00038 and 2024-04853, and he is grateful for the support.

\end{document}